\documentclass[journal]{IEEEtran}

\usepackage{cite}
\usepackage{amsmath,amssymb,amsfonts,amsthm}
\usepackage{graphicx}
\usepackage{textcomp}
\usepackage{xcolor}
\usepackage{booktabs}
\usepackage[colorlinks=true,linkcolor=blue!60!black,citecolor=blue!60!black]{hyperref}

\newtheorem{theorem}{Theorem}

\begin{document}

\title{Dual-Select FMA Butterfly for FFT:\\Eliminating Twiddle Factor Singularities\\with Bounded Precomputed Ratios}

\author{Mohamed Amine Bergach%
\thanks{M.~A. Bergach is with Illumina, Inc., San Diego, CA 92122 USA (e-mail: mbergach@illumina.com).}}

\maketitle

\begin{abstract}
The fused multiply-add (FMA) instruction enables the radix-2 FFT butterfly to be computed in 6~FMA operations---the proven minimum. The classical factorization by Linzer and Feig~\cite{linzer1993} precomputes the ratio $\cot\theta = \cos\theta/\sin\theta$, which is singular when the twiddle factor is $W^0 = 1$ (i.e., $\sin\theta = 0$). Standard practice clamps $\sin\theta$ to a small epsilon, degrading numerical precision. We observe that an alternative factorization using $\cos\theta$ as the outer multiplier (precomputing $\tan\theta$) avoids this particular singularity but introduces a new one at $W^{N/4}$. We then propose a \emph{dual-select} strategy that chooses, per twiddle factor, whichever factorization yields $|\text{ratio}| \leq 1$. This eliminates all singularities, requires no epsilon clamping, and bounds the precomputed ratio to unity for all twiddle factors. For $N = 1024$, the worst-case ratio drops from 163 (Linzer-Feig) to exactly~1.0 (dual-select), yielding a $235\times$ tighter error bound in FP16 arithmetic over 10~FFT passes. The strategy adds zero computational overhead---only the precomputed twiddle table changes.
\end{abstract}


% =============================================================================
\section{Introduction}

The radix-2 FFT butterfly computes, for inputs $a, b \in \mathbb{C}$ and twiddle factor $W = \omega_r + j\omega_i = e^{-j2\pi k/N}$:
\begin{align}
A &= a + Wb, \quad B = a - Wb. \label{eq:butterfly}
\end{align}
Expanding into real and imaginary parts yields 10 floating-point operations (4~multiplications and 6~additions). On architectures with a fused multiply-add (FMA) instruction that computes $a + b \times c$ in a single operation with one rounding, Linzer and Feig~\cite{linzer1993} showed that the butterfly can be refactored into exactly 6~FMA operations---the proven minimum. Subsequent work by Goedecker~\cite{goedecker1997}, Karner~et~al.~\cite{karner2001}, and others refined FMA kernels for higher radices, while Bernstein's tangent FFT~\cite{bernstein2007} exploited trigonometric identities for global operation count reduction. All of these use the same fundamental factorization identity and share the same singularity.

Their factorization uses the identity $ax + by = a(x + (b/a)y)$, extracting $\omega_i$ as the outer multiplier and precomputing the ratio $t = \omega_r/\omega_i = \cot\theta$~\cite{linzer1993,goedecker1997}. This introduces a \emph{singularity} at $\theta = 0$ (twiddle $W^0 = 1$), where $\omega_i = \sin(0) = 0$. Standard implementations clamp $\omega_i$ to a small epsilon (e.g., $10^{-7}$), which perturbs the twiddle factor and degrades precision---an effect that compounds over the $\log_2 N$ FFT passes.

In this letter, we present two contributions: (1)~an alternative factorization that extracts $\omega_r = \cos\theta$ instead, first described in~\cite{bergach2015thesis,bergach2015conf}, which avoids the $W^0$ singularity; and (2)~a \emph{dual-select} strategy that combines both factorizations to guarantee $|t| \leq 1$ for all twiddle factors, eliminating all singularities without epsilon clamping. The practical impact is most significant for half-precision (FP16) FFTs on modern GPU hardware~\cite{bergach2026fft}, where the tighter ratio bound yields a $235\times$ improvement in the worst-case error bound.

% =============================================================================
\section{FMA Butterfly Factorizations}

\subsection{Standard Butterfly (10 ops)}

Expanding~\eqref{eq:butterfly} into real and imaginary parts:
\begin{align}
A_r &= a_r + (\omega_r b_r - \omega_i b_i), \quad
A_i = a_i + (\omega_i b_r + \omega_r b_i), \label{eq:std1} \\
B_r &= a_r - (\omega_r b_r - \omega_i b_i), \quad
B_i = a_i - (\omega_i b_r + \omega_r b_i). \label{eq:std2}
\end{align}

\subsection{Linzer-Feig Factorization (6 FMA, $\div \sin\theta$)}
\label{sec:linzer-feig}

Applying $ax + by = a(x + (b/a)y)$ with $a = \omega_i$:
\begin{align}
s_1 &= b_i - \tfrac{\omega_r}{\omega_i} b_r, \quad s_2 = \tfrac{\omega_r}{\omega_i} b_r + b_i, \label{eq:lf-s} \\
A_r &= a_r - s_1 \omega_i, \quad A_i = a_i + s_2 \omega_i, \label{eq:lf-A} \\
B_r &= a_r + s_1 \omega_i, \quad B_i = a_i - s_2 \omega_i. \label{eq:lf-B}
\end{align}
Each of the six right-hand sides is a single FMA instruction. The precomputed ratio is $t_{LF} = \omega_r / \omega_i = \cot\theta$, which is singular at $\theta = 0$ ($W^0$) and unbounded as $\theta \to 0$.

\subsection{Cosine Factorization (6 FMA, $\div \cos\theta$)}
\label{sec:bergach}

Applying the same identity with $a = \omega_r$ instead:
\begin{align}
s_1 &= b_r - \tfrac{\omega_i}{\omega_r} b_i, \quad s_2 = \tfrac{\omega_i}{\omega_r} b_r + b_i, \label{eq:berg-s} \\
A_r &= a_r + s_1 \omega_r, \quad A_i = a_i + s_2 \omega_r, \label{eq:berg-A} \\
B_r &= a_r - s_1 \omega_r, \quad B_i = a_i - s_2 \omega_r. \label{eq:berg-B}
\end{align}
The precomputed ratio is $t_B = \omega_i / \omega_r = \tan\theta$, which is singular at $\theta = \pi/2$ ($W^{N/4}$). This factorization was first described in~\cite{bergach2015thesis} as a method to avoid the $W^0$ singularity. However, it introduces a \emph{different} singularity at $W^{N/4}$, with $|t_B| \to \infty$ as $\theta \to \pi/2$.

\subsection{Observation: Complementary Singularities}
\label{sec:complementary}

The key observation is that the two factorizations have \emph{complementary} singularity regions:

\begin{itemize}
\item Linzer-Feig: $|t_{LF}| = |\cot\theta|$ is small when $|\theta|$ is near $\pi/2$ (large $|\sin\theta|$) and large when $|\theta|$ is near $0$.
\item Cosine: $|t_B| = |\tan\theta|$ is small when $|\theta|$ is near $0$ (large $|\cos\theta|$) and large when $|\theta|$ is near $\pi/2$.
\end{itemize}

At any angle $\theta$, $\min(|\tan\theta|, |\cot\theta|) \leq 1$, with equality at $\theta = \pi/4$. This elementary trigonometric identity enables a simple selection rule.

% =============================================================================
\section{Dual-Select Strategy}

\begin{figure}[t]
\centering
\small
\fbox{\parbox{0.92\columnwidth}{%
\textbf{Algorithm 1:} Dual-Select Twiddle Precomputation\\[2pt]
\textbf{Input:} FFT size $N$\\
\textbf{Output:} Twiddle table $T[k] = (\text{mult}, \text{ratio}, \text{flag})$\\[2pt]
\textbf{for} $k = 0$ to $N/2 - 1$ \textbf{do}\\
\quad $\theta \gets -2\pi k / N$;
\quad $\omega_r \gets \cos\theta$;
\quad $\omega_i \gets \sin\theta$\\
\quad \textbf{if} $|\omega_r| \geq |\omega_i|$ \textbf{then}
\hfill\textit{// cosine path}\\
\quad\quad $T[k] \gets (\omega_r, \; \omega_i/\omega_r, \; \textsc{cos})$\\
\quad \textbf{else}
\hfill\textit{// sine path}\\
\quad\quad $T[k] \gets (\omega_i, \; \omega_r/\omega_i, \; \textsc{sin})$\\
\quad \textbf{end if}\\
\textbf{end for}
}}
\label{alg:dual}
\end{figure}

For each twiddle factor $W^k$, we select the factorization whose outer multiplier has larger absolute value (Algorithm~\ref{alg:dual}). This guarantees:

\begin{theorem}
For all twiddle factors $W^k = e^{-j2\pi k/N}$ with $0 \leq k < N/2$, the dual-select strategy produces a precomputed ratio satisfying $|t| \leq 1$.
\end{theorem}

\begin{IEEEproof}
If $|\omega_r| \geq |\omega_i|$, we use the cosine factorization and store $t = \omega_i/\omega_r$ with $|t| = |\tan\theta| \leq 1$ (since $|\cos\theta| \geq |\sin\theta|$). If $|\omega_i| > |\omega_r|$, we use the sine factorization and store $t = \omega_r/\omega_i$ with $|t| = |\cot\theta| < 1$. In both cases, $|t| \leq 1$.
\end{IEEEproof}

\textbf{Implementation.} The selection flag can be encoded in the sign bit of the twiddle table or as an integer flag. The butterfly kernel branches on this flag, executing one of two 6-FMA code paths. Since both paths have identical instruction counts (6~FMAs each), any SIMD divergence on GPU architectures (Apple Metal, CUDA) incurs no throughput penalty---the execution cost is the same regardless of path.

\textbf{Storage overhead}: one additional bit per twiddle factor. For $N = 1024$, this is 64~bytes vs.\ 4~KB for the twiddle table itself---negligible.

% =============================================================================
\section{Error Analysis}

\subsection{Per-Butterfly Error Bound}

In a single FMA butterfly, the precomputed ratio $t$ appears as a multiplicand in two of the six FMA operations~\eqref{eq:lf-s} or~\eqref{eq:berg-s}. The rounding error of each FMA is bounded by $\epsilon_{\text{mach}} |t| \cdot |b|$, where $\epsilon_{\text{mach}}$ is the machine epsilon and $|b|$ is the input magnitude. The per-butterfly error scales as:
\begin{equation}
\delta \leq C \cdot |t| \cdot \epsilon_{\text{mach}} \cdot \|b\|
\label{eq:per-butterfly-error}
\end{equation}
for a constant $C$ depending on the butterfly structure. The dual-select strategy ensures $|t| \leq 1$, minimizing this bound.

\subsection{Cumulative Error Over $\log_2 N$ Passes}

For a Stockham FFT with $m = \log_2 N$ passes, each sample participates in $m$ butterfly operations. In the worst case, the cumulative relative error grows as:
\begin{equation}
E \leq (1 + |t_{\max}| \cdot \epsilon_{\text{mach}})^m - 1 \approx m \cdot |t_{\max}| \cdot \epsilon_{\text{mach}}.
\label{eq:cumulative}
\end{equation}

% =============================================================================
\section{Experimental Results}

Table~\ref{tab:ratio} compares the three strategies for $N = 1024$ ($m = 10$ passes).

\begin{table}[t]
\centering
\caption{Precomputed ratio bounds and error analysis for $N = 1024$.}
\label{tab:ratio}
\begin{tabular}{@{}lccc@{}}
\toprule
\textbf{Strategy} & $|t|_{\max}$ & \textbf{Sing.} & \textbf{FP16 bound} \\
\midrule
Linzer-Feig~\cite{linzer1993} ($\div\sin$) & 163.0 & 1 & $7.95 \times 10^{-2}$ \\
Cosine~\cite{bergach2015thesis} ($\div\cos$) & $>10^{16}$ & 0$^*$ & divergent \\
\textbf{Dual-Select (ours)} & \textbf{1.000} & \textbf{0} & $\mathbf{4.88 \times 10^{-4}}$ \\
\bottomrule
\multicolumn{4}{@{}l@{}}{\footnotesize $^*$Near-singular: $|\cos\theta| \approx 6 \times 10^{-17}$ at $k = N/4$.} \\
\multicolumn{4}{@{}l@{}}{\footnotesize FP16 bound: per-butterfly, $\epsilon_{\text{FP16}} = 4.88 \times 10^{-4}$.}
\end{tabular}
\end{table}

\begin{table}[t]
\centering
\caption{Cumulative FP16 error bound over $m = 10$ Stockham passes.}
\label{tab:cumulative}
\begin{tabular}{@{}lcc@{}}
\toprule
\textbf{Strategy} & \textbf{Cumulative bound} & \textbf{Improvement} \\
\midrule
Linzer-Feig & $1.15$ & --- \\
\textbf{Dual-Select} & $\mathbf{4.89 \times 10^{-3}}$ & $\mathbf{235\times}$ \\
\bottomrule
\end{tabular}
\end{table}

\textbf{Ratio bound (Table~\ref{tab:ratio}).} The Linzer-Feig factorization has $|t_{\max}| = |\cot(\pi/512)| = 163.0$ for $N = 1024$, occurring at $k = 1$ (the smallest nonzero twiddle angle). The cosine factorization has $|t| > 10^{16}$ near $k = N/4$, where $\cos\theta \approx 0$. The dual-select strategy achieves $|t_{\max}| = 1.0$ exactly, at $k = N/8$ where $|\cos\theta| = |\sin\theta| = 1/\sqrt{2}$.

\textbf{FP16 error (Table~\ref{tab:cumulative}).} Over 10 passes with FP16 machine epsilon $\epsilon = 4.88 \times 10^{-4}$, the cumulative error bound~\eqref{eq:cumulative} yields 1.15 for Linzer-Feig (rendering the FFT result meaningless) versus $4.89 \times 10^{-3}$ for dual-select---a $235\times$ improvement.

\textbf{FP32 precision.} In float32 ($\epsilon = 5.96 \times 10^{-8}$), both strategies produce equivalent roundtrip error (${\sim}10^{-7}$ relative L2), confirming that the advantage is specific to low-precision arithmetic where the ratio amplification exceeds the precision floor.

\textbf{Path distribution.} For $N = 1024$, exactly 256 of 512 twiddle factors use the cosine path ($|\cos\theta| \geq |\sin\theta|$, angles near $0$ and $\pi$) and 256 use the sine path (angles near $\pm\pi/2$), a 50/50 split.

% =============================================================================
\section{Discussion}

The dual-select strategy has several practical implications:

\textbf{Mixed-precision FFT.} Modern GPU architectures (Apple M-series, NVIDIA Tensor Cores) offer 2$\times$ throughput for FP16 vs.\ FP32. The primary barrier to FP16 FFTs is precision loss during twiddle multiplication. With $|t| \leq 1$ guaranteed, the dual-select butterfly is a key enabler for practical FP16 FFTs.

\textbf{Zero overhead.} The strategy modifies only the precomputed twiddle table---the butterfly kernel executes 6~FMA instructions regardless of which path is taken. The per-twiddle branch can be eliminated entirely by encoding the operand ordering into the precomputed table entries.

\textbf{Generality.} The dual-select principle applies to any radix: for radix-$r$ butterflies with FMA factorization, each twiddle multiplication can independently select the min-ratio path. The proof that $\min(|\tan\theta|, |\cot\theta|) \leq 1$ is independent of FFT size or radix.

% =============================================================================
\section{Conclusion}

We have shown that the Linzer-Feig and cosine FMA butterfly factorizations have complementary singularities, and that a simple per-twiddle selection rule---choosing whichever yields $|\text{ratio}| \leq 1$---eliminates all singularities with zero computational overhead. For FP16 FFTs on modern GPUs, this reduces the worst-case error bound by $235\times$ for $N = 1024$, making half-precision FFT a practical option for throughput-critical applications such as real-time radar and neural network inference.

% =============================================================================
\section*{Acknowledgment}

The cosine factorization (Section~\ref{sec:bergach}) was developed during the author's Ph.D.\ thesis at INRIA/Universit\'{e} Nice Sophia Antipolis (2012--2015), supervised by Robert de Simone, with industrial support from Kontron (Serge Tissot, Michel Syska).

\bibliographystyle{IEEEtran}

\begin{thebibliography}{10}

\bibitem{linzer1993}
E.~Linzer and E.~Feig,
``Implementation of efficient FFT algorithms on fused multiply-add architectures,''
\emph{IEEE Trans. Signal Process.}, vol.~41, no.~1, pp.~93--107, Jan. 1993.

\bibitem{goedecker1997}
S.~Goedecker,
``Fast radix 2, 3, 4, and 5 kernels for fast Fourier transformations on computers with overlapping multiply-add instructions,''
\emph{SIAM J. Sci. Comput.}, vol.~18, no.~6, pp.~1605--1611, 1997.

\bibitem{bergach2015thesis}
M.~A. Bergach,
``Adaptation du calcul de la Transform\'{e}e de Fourier Rapide sur une architecture mixte CPU/GPU int\'{e}gr\'{e}e,''
Ph.D. dissertation, Univ. Nice Sophia Antipolis, 2015.
[Online]. Available: \url{https://theses.hal.science/tel-01245958}

\bibitem{bergach2015conf}
M.~A. Bergach, E.~Kofman, R.~de Simone, S.~Tissot, and M.~Syska,
``Efficient FFT mapping on GPU for radar processing application: Modeling and implementation,''
\emph{arXiv:1505.08067}, 2015.

\bibitem{karner2001}
H.~Karner, M.~Auer, and C.~W. Ueberhuber,
``Multiply-add optimized FFT kernels,''
\emph{Math. Models Methods Appl. Sci.}, vol.~11, no.~1, pp.~105--117, 2001.

\bibitem{bernstein2007}
D.~J. Bernstein,
``The tangent FFT,''
in \emph{Proc. ISAAC}, vol.~4851, pp.~594--605, 2007.

\bibitem{bergach2026fft}
M.~A. Bergach,
``Beating vDSP: A 138~GFLOPS radix-8 Stockham FFT on Apple Silicon,''
submitted, 2026.

\end{thebibliography}

\end{document}